\documentclass[letterpaper, 10 pt, conference]{ieeeconf}  

\IEEEoverridecommandlockouts                              

\usepackage{cite}
\usepackage{color}
\usepackage{xcolor}
\usepackage{mathrsfs}
\usepackage{verbatim}
\usepackage{graphicx}
\usepackage{bbold}
\usepackage[draft]{todonotes}   
\usepackage{fancyhdr}
\usepackage{stfloats}
\usepackage{subfigure}

\usepackage[font=footnotesize,labelfont=bf]{caption}

\usepackage{graphicx}
\usepackage{tabularx}

\usepackage[cmex10]{amsmath}
\usepackage{mathtools, cuted}
\usepackage{array}
\usepackage{amsfonts}
\usepackage[ruled, vlined, algo2e]{algorithm2e}
\usepackage{url}
\usepackage{enumerate}
\usepackage{amsthm}
\usepackage{latexsym}
\usepackage{amssymb} 
\usepackage{amsmath} 
\usepackage{graphics,graphicx,color}
\usepackage{float}

\newcommand{\R}[2]{R^{(#1)}_{#2}}
\newcommand{\M}[2]{M^{(#1)}_{#2}}

\newcommand{\p}[2]{\textit{P}_{\eta_{#1}}({#2})}
\newcommand{\pp}[1]{\textit{P}_{\eta_{#1}}}

\newcommand{\E}{\mathbb{E}}

\newcommand{\RR}{\mathbb{R}}

\renewcommand{\b}[1]{\mathbf{#1}}
\newcommand{\TX}{\b{\Psi}^{x}}
\newcommand{\TU}{\b{\Psi}^{u}}
\newcommand{\Tx}[1]{{\Psi}^{x}_{#1}}
\newcommand{\Tu}[1]{\Psi^{u}_{#1}}

\newcommand{\sat}{\mathrm{sat}}

\theoremstyle{plain}

\newcommand{\lemref}[1]{(Lem.\ref{#1})}
\newcommand{\thmref}[1]{(Thm.\ref{#1})}

\newcommand{\secref}[1]{(Sec.\ref{#1})}

\newcommand{\figref}[1]{Fig. \ref{#1}}
\newcommand{\defref}[1]{(Def. \ref{#1})}

\newcommand{\CLM}[1]{\b{\Psi}^{#1}}

\newcommand{\clm}[2]{\Psi^{#1}_{#2}}
\newcommand{\feas}{\b{\overline{\Phi}}[\mathbf{F}]}
\newcommand{\SL}[2]{\mathrm{SL}(#1,#2)}
\newcommand{\set}[1]{[#1]}
\newcommand{\causal}[2]{\mathcal{C}(\ell^{#1},\ell^{#2})}
\newcommand{\scausal}[2]{\mathcal{C}_s(\ell^{#1},\ell^{#2})}
\newcommand{\lcausal}[2]{\mathcal{L}\mathcal{C}(\ell^{#1},\ell^{#2})}
\newcommand{\lscausal}[2]{\mathcal{L}\mathcal{C}_s(\ell^{#1},\ell^{#2})}

\theoremstyle{definition}
\newtheorem{coro}{Corollary}[section]
\newtheorem{thm}{Theorem}[section]
\newtheorem{lem}{Lemma}[section] 

\newtheorem{defn}[thm]{Definition} 
\newtheorem{asp}[coro]{Assumption}
\newtheorem{rem}[coro]{Remark}

\usepackage{algorithm,algcompatible,amsmath}
\algnewcommand\INPUT{\item[\textbf{Input:}]}%
\algnewcommand\OUTPUT{\item[\textbf{Output:}]}%

\title{\LARGE \bf
Achieving Performance and Safety in Large Scale Systems with Saturation using a Nonlinear System Level Synthesis Approach
}

\author{Jing Yu* and Dimitar Ho*
\thanks{* The authors contributed equally to this work.}
\thanks{Jing Yu and Dimitar Ho are with the Department of Computing and Mathematical Sciences, California Institute of Technology, Pasadena, CA. }}%

\renewcommand{\headrulewidth}{0pt} 
\begin{document}
\setlength{\abovedisplayskip}{3 pt}
\setlength{\belowdisplayskip}{3 pt}
\maketitle
\thispagestyle{fancy}
\pagestyle{fancy}
\pagestyle{fancy}
\fancyhead{} 
\renewcommand{\headrulewidth}{0pt} 
\fancyfoot{} 


\begin{abstract}

We present a novel class of nonlinear controllers that interpolates among differently behaving linear controllers as a case study for recently proposed Linear and Nonlinear System Level Synthesis framework. The structure of the nonlinear controller allows for simultaneously satisfying performance and safety objectives defined for small- and large-disturbance regimes. The proposed controller is distributed, handles delays, sparse actuation, and localizes disturbances. We show our nonlinear controller always outperforms its linear counterpart for constrained LQR problems. We further demonstrate the anti-windup property of an augmented control strategy based on the proposed controller for saturated systems via simulation. 
\end{abstract}

\section{Introduction}
In this paper, we propose a novel offline distributed nonlinear controller synthesis procedure that outperforms any optimal linear controller for the constrained LQR problem \cite{maciejowski2002predictive, zhang2013distributed,li2011disturbance,SLSChenAnderson2019}. With simple augmentation, our controller controller inherently prevents windup-instabilities in saturated linear systems which are traditionally mitigated by additional anti-windup design \cite{kothare1994unified, kolmanovsky2014reference,hu2008anti}. Another significant advantage of the approach, is that despite being a nonlinear synthesis method it naturally enjoys the same benefits as the linear system level approach introduced in \cite{anderson2019system}, which allows for localized controller implementation, making it scalable to large networks. 

 Our work is based on \cite{NLSLSHODoyle2019}, which describes the system-level characterization of the closed loops of general nonlinear discrete-time systems. Moreover, \cite{NLSLSHODoyle2019} introduces a simple universal control structure, called a system level controller, that has the capacity to stabilize any nonlinear system if parametrized with the according closed loop maps. In this paper, we will show that just using a very special case of the framework presented in \cite{NLSLSHODoyle2019} provides new promising tools for control design. In particular, we will illustrate how a simple projection nonlinearity can become a powerful tool for solving the problems described above.

The rest of the paper starts with a review on nonlinear System Level Synthesis (NLSLS) in \secref{sec:nlsls}. The proposed nonlinear controller is introduced in \secref{sec:blend}, followed by \secref{sec:lqr} where the constrained LQR problem is discussed. We show in \secref{sec:ls} that the proposed nonlinear controller can be augmented for natural anti-windup properties and therefore allow for large-scale distributed anti-windup design. Numerical simulation in \secref{sec:sim} corroborates the presented theory.

\section{Preliminaries and Notation}
 We will define $\ell^n$ to be the space of sequences of vectors in $\RR^n$.  Sequences of vectors will be denoted by small bold letters $\b{x}:=(x_t)^{\infty}_{t=0}$  unless otherwise specified.
 Occasionally, we will define sequences explicitly with the tuple notation $\b{x} := (x_0,x_1,\dots)$ and $x^j_t$ denotes the $j$th element of vector $x_t$. 
 We use the $x_{i:j}$ to refer to the truncation of a sequence $\b{x}$ to the tuple $(x_{i},x_{i+1},\dots, x_{j})$. Furthermore, we will adopt $|x|$ and $|A|$ for $x \in \RR^n$ and $A \in \RR^{n \times n}$ as the vector  $\infty$-norm and induced $\infty$-norm on $\RR^n$, respectively.
We reserve $\|\cdot\|_{p}$ to refer to the norm and induced norm over vector sequence space $\ell _p$:
 \begin{equation*}
    \|\b{x}\|_p:=\left(\sum_{k = 0}^{\infty} \lvert x_k\lvert^p \right)^{\frac{1}{p}} \qquad \|\b{x}\|_{\infty} := \sup_{k\geq 0} |x_k|.
\end{equation*}
 Finally, the set of positive integers ranging from 1 to $N$ will be denoted as $\set{N}$.

 \subsection{Operators}
 Operators that maps between sequence spaces will be denoted in bold capital letters $\b{T}:\ell^n \rightarrow \ell^k $. Similar to the sequence of vectors, we write $\b{T}:= \{T_t\}_{t = 0}^{\infty}$ with its component functions $T_t:\ell^n \rightarrow \RR^k $. 
 An operator $\b{T}$ will be called \textit{causal} if for any pair of input $\b{x}$ and its corresponding output $\b{y} = \b{T}(\b{x})$, the output $y_t$ does not depend on future input sequence $x_{t+k}, k\geq 1$. More precisely, a causal operator $\b{T}$ is fully characterized by its component functions $T_t:\RR^{n \times (t+1)} \rightarrow \RR^k $ such that: 
 $$\b{T}(\b{x})=(T_0(x_0), T_1(x_1,x_0),T_2(x_2,x_1,x_0), \dots).$$
 Note that every component function $T_t$ of a causal operator $\b{T}$ has $t+1$ arguments which are populated in reverse-chronological order. If in addition , component functions $T_t$ satisfy $T_t(x_{t:0}) = T_t(0,x_{t-1:0})$, then $\b{T}$ will be called \textit{strictly causal}.

 We define the space of all causal and strictly causal operators that maps $\ell^n \rightarrow \ell^p$ as $\causal{n}{p}$ and $\scausal{n}{p}$, respectively. Moreover, let the space of all linear causal and strictly causal operators be denoted as $\lcausal{n}{p} \subset \causal{n}{p}$ and $\lscausal{n}{p} \subset \scausal{n}{p}$. Occasionally, for two operators with matching domains such as $\b{A} \in \causal{n}{p}$ and $\b{B}\in \causal{n}{q}$, we denote the composite operator $(\b{A},\b{B}) \in \mathcal{C}(\ell^n,\ell^p \times \ell^q)$ as $(\b{A},\b{B}):\b{x} \mapsto (\b{A}(\b{x}),\b{B}(\b{x}))$.

 \subsection{$\ell_p$ Stability}
 Let the vector sequence space $\ell^n_p \subset \ell^n$ be defined as:
 $$\ell^n_p := \{ \b{x}\in\ell^n\, | \, \|\b{x}\|_p < \infty \}.$$ 
 We define stability for causal operators as follows:
 \begin{defn}[$\ell_p$ Stability]
    An operator $\b{T} \in \causal{n}{m}$ is said to be $\ell_p$-stable, if 
    $$\b{T}(\b{x}) \in \ell^m_p \,\,\,\, \text{for all} \,\, \b{x} \in \ell^n_p.$$
    Further, if there exist two scalars $\gamma$, $\beta \geq 0$ such that for all $\b{x} \in \ell^n_p$, we have 
    $$\|\b{T}(\b{x})\|_p \leq \gamma \|\b{x}\|_p + \beta,$$
    then $\b{T}$ is finite gain $\ell_p$-stable.
 \end{defn}

\section{An Overview of the Nonlinear System Level Approach}\label{sec:nlsls}
 This section will focus on introducing the notion of closed loop maps as causal operators with respect to a general nonlinear causal system. Moreover, we summarize necessary and sufficient conditions for operators to be closed loop maps and how they can be realized by a dynamic controller.

\subsection{Closed Loop Maps of Nonlinear Systems}
Consider a discrete-time nonlinear system with additive disturbances
\begin{equation}
    \label{eq:sys}
    x_{t} = f(x_{t-1},u_{t-1}) +w_t,
\end{equation}
where $x_t\in \mathbb{R}^n$, $u_t\in \mathbb{R}^m$, $w_t\in \mathbb{R}^n$ and $f:\RR^n \times \RR^{m} \rightarrow \RR^n$ with $f(0,0) = 0$ and $x_0 = w_0$. Let $\b{F}(\b{x},\b{u}): \ell^n \times \ell^m \to \ell^n$ be the strictly causal operator representation of the function $f$ such that $\b{F}(\b{x},\b{u}):= (0,f(x_0,u_0),f(x_1,u_1),\dots)$.
Assume that $w_t$ can not be measured and that $u_t$ is generated by some causal controller $\b{K} \in \causal{n}{m}$ such that $ u_t = K_t\left(x_{t:0} \right)$. An equivalent operator form of the dynamics \eqref{eq:sys} is
\begin{subequations}
    \label{eq:dyna}
    \begin{align}
        \b{x} &= \b{F}(\b{x},\b{u}) + \b{w} \label{eq:open-loop} \\
        \b{u} &= \b{K}(\b{x}) \label{eq:control}.
    \end{align}
\end{subequations}
For a fixed disturbance sequence $\b{w}$, the dynamics \eqref{eq:sys} produces unique closed loop trajectories for state $\b{x}$ and input $\b{u}$. Therefore, given a fixed $\b{K}$, the dynamics induce a causal map from $\b{w}$ to $(\b{x}$, $\b{u})$ and we will call the corresponding operators disturbance-to-state and disturbance-to-input \textit{closed loop map}, respectively.
\begin{defn}[Closed Loop Maps]\label{def:clm}
    Define $\b{\Phi}[\b{F},\b{K}] \in \mathcal{C}(\ell^n,\ell^n \times \ell^m)$ as the operator that maps $\b{w}$ to the corresponding response $(\b{x},\b{u})$ according to the closed-loop dynamics \eqref{eq:dyna}. We call $\b{\Phi}[\b{F},\b{K}]$ the \textit{closed loop maps} (CLMs) of \eqref{eq:dyna}. Moreover we will refer to the partial maps $\b{w} \rightarrow \b{x}$ and $\b{w} \rightarrow \b{u}$ with $\b{\Phi}^{\b{x}}[\b{F},\b{K}]$ and $\b{\Phi}^{\b{u}}[\b{F},\b{K}]$, respectively. 
\end{defn}
Without specifying a controller $\b{K}$, one could alternatively consider the \textit{realizable} CLMs of \eqref{eq:open-loop} for \textit{some} causal controller $\b{K'}$. We call a composite operator $\b{\Psi} = (\b{\Psi}^{\b{x}}, \CLM{u}) \in \mathcal{C}(\ell^n,\ell^n \times \ell^m)$ that maps $\b{w} \mapsto \left(\b{\Psi}^{\b{x}}(\b{w}), \b{\Psi}^{\b{u}}(\b{w}) \right)$  \textit{realizable CLMs} for open-loop dynamic \eqref{eq:open-loop} if there exists a so-called \textit{realizing} controller $\b{K}'$ such that $\b{\Psi} = \b{\Phi}[\b{F},\b{K}']$. With this notion of realizable CLMs of an open-loop dynamics, we define the space of all realizable CLMs:
\begin{defn}[Space of Realizable CLMs]
    Given an open-loop dynamics \eqref{eq:open-loop}, the set of all feasible closed loop maps $\feas \in \mathcal{C}(\ell^n,\ell^n \times \ell^m)$ for open loop \eqref{eq:open-loop} is defined as:
    $$\feas:= \{\b{\Psi}| \exists \b{K} \in \causal{n}{m} \text{ s.t. }\b{\Psi}= \b{\Phi}[\b{F},\b{K}] \}.$$
\end{defn}
The following theorem characterizes the space of realizable CLMs for a given open loop:

\begin{thm}[Characterization of CLMs \cite{NLSLSHODoyle2019}]
    \label{thm:sufnec}
    A composite operator $\b{\Psi}=(\CLM{x},\CLM{u}) \in \mathcal{C}(\ell^n,\ell^n \times \ell^m)$ are realizable CLMs of the open loop \eqref{eq:open-loop} if and only if they satisfy the operator equation
\begin{align}
\label{eq:feasnl} \CLM{x} &= \b{F}(\b{\Psi}) + \b{I}.
\end{align}
Moreover, for any operators $\CLM{}$ satisfying \eqref{eq:feasnl}, the inverse $(\CLM{x})^{-1}$ exists, is a causal operator, and $\b{K} = \CLM{u}(\CLM{x})^{-1}$ is a realizing controller for CLMs $\CLM{}$. If $\CLM{u}$ is surjective, then $\b{K}$ is unique.
\end{thm}

\subsection{System Level Implementations}
Aside from the technical assumption on the codomain of $\CLM{u}$, \thmref{thm:sufnec} states that there is a one-to-one relation between CLMs $(\CLM{x},\CLM{u})$ and their realizing controllers $\b{K} = \CLM{u}(\CLM{x})^{-1}$. Nevertheless, different implementations of $\b{K}$ need to be distinguished: despite realizing the same CLMs with respect to the trajectory $(\b{w},\b{x},\b{u})$, they do not give the same closed loop behavior once we add additional perturbations to the system. We will denote the following realization of $\b{K} = \CLM{u}(\CLM{x})^{-1}$ as the \textit{System Level} (SL)-implementation of $\b{K}$:
\begin{defn}[SL Implementation]\label{def:SL}
Given a composite operator $\b{\Psi}=(\CLM{x},\CLM{u}) \in \mathcal{C}(\ell^n,\ell^n \times \ell^m)$ satisfying \eqref{eq:feasnl}, the realizing controller $\b{K} = \CLM{u}(\CLM{x})^{-1}$ can be implemented as follows :
\begin{subequations}
	\label{eq:SLeq}
    \begin{align}
	    u_t &= \Tu{t}(\hat{w}_{t:0})\\
	    \hat{w}_{t+1} &= x_{t+1}-\Tx{t+1}(0,\hat{w}_{t:0})
	\end{align}
\end{subequations}
for $t = 0,1,\dots$, where $\hat{w}$ denotes the internal state of the controller with initial condition $\hat{w}_0 = x_0$. We will write $\b{K}=\SL{\TX}{\TU}$ to underscore that the controller $\b{K}=  \CLM{u}(\CLM{x})^{-1}$ is implemented in this fashion.
\end{defn}

Consider the closed loop of \eqref{eq:sys} and controller $\b{K}=\SL{\TX}{\TU}$ perturbed by additional noise $\b{v}$ and input disturbance $\b{d}$ such that:
\begin{subequations}
\label{eq:sysperturb} 
\begin{align}
    	x_{t} &= f(x_{t-1},u_{t-1})+w_t\\
	    u_t &= \Tu{t}(\hat{w}_{t:0})+d_t\\
	    \hat{w}_t &= x_t-\Tx{t}(0,\hat{w}_{t-1:1}) + v_t.
\end{align}
\end{subequations}
 \begin{thm}[Internal Stability of Closed Loop\cite{NLSLSHODoyle2019}]\label{thm:stability}
If $f$ is uniformly continuous and the operator $(\CLM{x},\CLM{u})$  is $\ell^{n+m}_p$-stable (or $\ell^{n+m}_p$ finite gain-stable) CLMs of \eqref{eq:sys} , then the closed loop dynamics \eqref{eq:sysperturb} are $\ell^{n+m}_p$-stable (or finite gain $\ell^{n+m}_p$-stable) with respect to the perturbation $(\b{w},\b{d},\b{v})$. 
\end{thm}
%
%
\subsection{Relation to Linear System Level Approach}
If we restrict the previous analysis to linear time-invariant (LTI) systems and controllers, we recover the results of \cite{anderson2019system} for the state-feedback case.
If the open-loop dynamics now is $x_{t} = Ax_{t-1} + Bu_{t-1} + w_t$ and $\b{K}$ is an LTI operator, then the corresponding linear CLMs are LTI as well, whose component functions can be written as:
\begin{subequations}
\label{eq:linear-maps}
\begin{align}
    \clm{x}{t}(w_{t:0}) & = \sum^{t+1}_{k=1}R_k w_{t+1-k}\\
    \clm{u}{t}(w_{t:0}) & = \sum^{t+1}_{k=1}M_k w_{t+1-k},
\end{align}
\end{subequations}
where $R_k \in \RR^{n \times n}$, $M_k \in \RR^{m \times n}$ for $k = 1,2,\dots,t+1$ are called associated matrices of the component functions $\Psi^x_t$, $\Psi^u_t$ of $\CLM{x}$ and $\CLM{u}$. Moreover, the corresponding CLMs characterization condition \eqref{eq:feasnl} reduces to the affine constraint on the matrices $R_k$, $M_k$ which coincides with the linear System Level Synthesis (SLS) feasibility conditions derived in \cite{anderson2019system}. In particular, if we further 
restrict $\CLM{}$ to have Finite Impulse Response (FIR) with horizon $T$, \textit{i.e.}, component functions $\Psi^x_t$ and $\Psi^u_t$ only depend on the past $\min\{T,t+1\}$ inputs, then \eqref{eq:linear-maps} becomes:
\begin{subequations}
    \begin{align}
        \clm{x}{t}(w_{t:0}) & = \sum^{\min\{t+1,T\}}_{k=1}R_k w_{t+1-k}\\
        \clm{u}{t}(w_{t:0}) & = \sum^{\min\{t+1,T\}}_{k=1}M_k w_{t+1-k},
    \end{align}
    \end{subequations}
The CLMs characterization \eqref{eq:feasnl} in this the FIR LTI case reduces to the following conditions on the associated matrices $R_k$, $M_k$ for $k = 1,\dots,T-1$: 
\begin{subequations}
\label{eq:lin-feaibility}
\begin{align}
    &R_{1} = I \\
    &R_{k+1} = A R_k + B M_k \\
    &A R_{T} + B M_{T} = 0.
\end{align}
\end{subequations}
and $\b{K} = \SL{\CLM{x}}{\CLM{u}}$ results in the implementation below, which also coincides with that of \cite{anderson2019system}:
\begin{align*}
    u_t &= \sum^{\min\{t+1,T\}}_{k=1}M_{k}\hat{w}_{t+1-k}    \\
    \hat{w}_{t+1} &= x_{t+1} - \sum^{\min\{t+2,T\}}_{k=2}R_{k}\hat{w}_{t+2-k},
\end{align*}
for all $k = 0,1,\dots$ with $\hat{w}_0 = x_0$.

\section{Nonlinear Blending of Linear System Level Controllers}
\label{sec:blend}

As introduced in the previous section, system level controllers defined in \defref{def:SL} can implement arbitrary CLMs for nonlinear systems of the form \eqref{eq:sys}. The results in \cite{NLSLSHODoyle2019} motivate a new approach for nonlinear control synthesis: Searching for stable operators $\TX$, $\TU$ that satisfy \eqref{eq:feasnl} and constructing a corresponding system level controller $\SL{\TX}{\TU}$ by \defref{def:SL}.

It is conceivable that the generality of this approach could lead to an entirely new direction of nonlinear dynamic control methods. Serving as a first step towards exploring the potential of this new perspective, the remainder of this paper focuses on a subset of nonlinear system level controllers $\SL{\TX}{\TU}$ that proves particularly useful for controlling large-scale linear systems subject to state/input constraints and input saturation.

%

In particular we will restrict ourselves to the class of controllers $\SL{\TX}{\TU}$ where $\TX$ and $\TU$ are structured as
    \begin{align}
        \label{eq:n-blend}
	\clm{x}{t}(\cdot) &= \sum^{N}_{i=1} \sum^{\min\{T,t+1\}}_{k=1}R^{(i)}_{k}  (\pp{i}-\pp{i-1})(w_{t+1-k})\nonumber\\
	\clm{u}{t}(\cdot) &= \sum^{N}_{i=1} \sum^{\min\{T,t+1\}}_{k=1}M^{(i)}_{k}(\pp{i}-\pp{i-1})(w_{t+1-k}).
    \end{align}
We choose $\eta_N\geq\eta_{N-1}\geq \dots \geq \eta_0 =  0$ and  the operator $\p{i}{\cdot}:\RR^{n} \rightarrow \RR^n$ as any nonlinear function with a projection-like property defined for parameter $\eta_i$. $R^{(i)}_k \in \RR^{n\times n}$, $M^{(i)}_k \in \RR^{m\times n}$ are matrices associated with linear FIR CLMs $\CLM{x,i}$,$\CLM{u,i}$, $i \in \set{N}$ with FIR horizon $T$ for a linear system of interests: 
\begin{align} \label{eq:syslin}x_{t} =Ax_{t-1} + Bu_{t-1} + w_{t}, \end{align}
with $x_t \in \RR^n$,$w_t \in \RR^n$, $u \in \RR^m$ such that for each $i \in \set{N}$, $\CLM{x,i}$,$\CLM{u,i}$ satisfies \eqref{eq:lin-feaibility}.
 Concretely, we consider two specific nonlinear projections:
 
 \begin{defn}[Saturation Projection]\label{def:sat}
     Let vector $w = [w^1, \dots, w^n]^T\in \RR^n$. The saturation projection is an element-wise projection:
     \begin{align}
 	\p{}{w}:= \begin{bmatrix}
 	\sat(w^1,\eta)\\ \vdots\\
 		\sat(w^n,\eta)
 	\end{bmatrix}
 	\end{align}
 	where $\sat(w,\eta) = \text{sign}(w) \max\{|w|,\eta\}$.
 \end{defn}

\begin{defn}[Radial Projection] \label{def:radial}
The radial projection is defined as:\begin{align}
    \p{}{w}:= \frac{\sat(|w|/\eta,1)}{|w|/\eta} w
\end{align}
\end{defn}
Unless otherwise specified, the results derived in the rest of the paper hold for both projections.
\begin{rem}
For $n=1$, radial projection and saturation projection coincide with each other. The radial and saturation projection operator act as the identity whenever $|w|\leq \eta$. Otherwise, the radial projection rescales $w$ such that $|\p{}{w}|= \eta$ whereas the saturation projection performs element-wise radial projection.
\end{rem}

The proposed nonlinear controller $\SL{\CLM{x}}{\CLM{u}}$ can be thought of as a \textit{nonlinear blend} of the \textit{linear} FIR controllers $\SL{\CLM{x,i}}{\CLM{u,i}}$, $i \in \set{N}$. Although the nonlinear operator $\CLM{x}$, $\CLM{u}$ differs from its linear components $\CLM{x,i}$, $\CLM{u,i}$ only by the static nonlinear function $\p{i}{w}$, the upcoming sections will demonstrate that this simple additional nonlinearity proves surprisingly useful. In particular, $\eta_i$'s separate any disturbance $w_t$ into $N$ zones such that for each $i$th linear controller $\SL{\CLM{x,i}}{\CLM{u,i}}$, only the portion of $w_t$ that "falls" between $\eta_i$ and $\eta_{i-1}$ is acted upon. Intuitively, one could choose different behaviors for various portions of the disturbance signal, specifying either performance or safety properties. The explicit expression of the dynamic controller $\SL{\TX}{\TU}$ with CLMs defined in \eqref{eq:n-blend} is:
\begin{align}
    u_t &= \sum^{N}_{i=1} \sum^{\min\{T,t+1\}}_{k=1}M^{(i)}_{k}(\pp{i}-\pp{i-1})(\hat{w}_{t+1-k})\nonumber\\
    \hat{w}_{t+1} &= x_{t+1} - \sum^{N}_{i=1} \sum^{\min\{T,t+2\}}_{k=2}R^{(i)}_{k}  (\pp{i}-\pp{i-1})(\hat{w}_{t+2-k}), \nonumber
\end{align}
with $k = 0,1,\dots$, and $\hat{w}_0 = x_0$.

For ease of exposition, we focus on the two-zone case of the proposed controller $\SL{\CLM{x}}{\CLM{u}}$ though all the analysis extends naturally to the $N$-zone case. 
Thus, \eqref{eq:n-blend} simplify to
\begin{align}
    \clm{x}{t}(w_{t:0}) = \sum^{\min\{T,t+1\}}_{k=1} &\R1{k} \p{1}{w_{t+1-k}} + \nonumber \\
    & \R2{k}(\p2{w_{t+1-k}} - \p1{w_{t+1-k}}) \nonumber \\
    \clm{u}{t}(w_{t:0}) = \sum^{\min\{T,t+1\}}_{k=1} &\M1{k} \p{1}{w_{t+1-k}} +\nonumber \\
    & \M2{k}(\p2{w_{t+1-k}} - \p1{w_{t+1-k}}).\label{eq:NL-map}
\end{align}
Note that system level controller $\SL{\CLM{x}}{\CLM{u}}$ of the two-zone CLM is internally stabilizing and achieves the two-zone CLM behavior for \eqref{eq:syslin}  as long as $\|\b{w}\|_{\infty} \leq \eta_2$. 

 In the remainder of this paper we will explore the consequence of this blending technique for distributed control design with respect to input saturation and state constraints in linear systems. we show that the simple nonlinearity in \eqref{eq:NL-map} offers a variety of advantages over linear controllers.

\section{A General Framework for Constrained LQR}
\label{sec:lqr}
We present a novel synthesis procedure for a class of constrained LQR problems using the proposed SL controller with CLMs \eqref{eq:NL-map}. In particular, we will show that the synthesized  nonlinear blending system level controller is guaranteed to outperform any linear controller for the class of constrained LQR problems to be discussed. Additionally, we comment on how structural constraints for large-scale systems such as delay, actuation sparsity, and localization can be easily accommodated. 

Consider a control problem where we wish to minimize an \textit{average} LQR cost, but also want that the closed loop meets certain safety guarantees against a set of rare yet possible \textit{worst-case} disturbances. Ideally, we would like to synthesize a controller that can guarantee the necessary safety constraints without too much loss in performance compared to the unconstrained LQR controller. We will phrase this design goal as the following constrained LQR problem:
\begin{subequations}
\label{eq:const-lqr}
\begin{align}
\label{eq:cost-lqr}&\min_{\b{K}}\lim_{T \rightarrow \infty}\frac1T \sum_{t=1}^T \mathbb{E}_{w_t^i \sim p(w)}[\mathcal{J}(x_t,u_t) ] \\
&s.t.\,\,\,\,\,\,\,\,\,\,\,\, x_{t} = Ax_{t-1}+Bu_{t-1} + w_{t}\label{eq:lin-dyn}\\
 &\qquad  \,\,\,\,\,\,\,\,\,\,\,\,\,\,\,u_t = K_t(x_{t:0})\\
 \label{eq:XUconstlqr}& \qquad \forall \b{w}: ||\b{w}||_{\infty}\leq \eta_{max}:  \\
&\,\,\,\,\,\,\,\,\,\,\,\,\,\qquad  \sup_k |x_k|\leq x_{max}, \quad \sup_k |u_k| \leq u_{max}  \nonumber
\end{align}
\end{subequations}
where $\mathcal{J}$ abbreviates the quadratic stage cost $\mathcal{J}(x,u) = x^TQx+uPu$ with $Q$,$P \succ 0$. We will assume that the disturbance is stochastic but bounded such that $\|\b{w}\|_{\infty} \leq \eta_{max}$ with known distribution which satisfies the following
\begin{asp}
\label{assum:distribution}
Disturbance $w_t^i$ are i.i.d. drawn from the scalar centered distribution $p(w)$ and uncorrelated in time $t$ and coordinate $i$.
\end{asp}

We can equivalently phrase the optimal control problem \eqref{eq:const-lqr} in terms of closed loop maps as defined in \secref{sec:nlsls}. Recalling \defref{def:clm}, the optimal control problem \eqref{eq:const-lqr} can be described as an optimization over the set of feasible CLMs $(\CLM{x},\CLM{u})\in \overline{\b{\Phi}}(\b{A}\b{x}+\b{B}\b{u})$ and by using the characterization \thmref{thm:sufnec} we obtain:
\begin{subequations}
\label{eq:const-lqrclm}
\begin{align}
\label{eq:cost-lqrclm-cost2}&\min_{\CLM{x},\CLM{u}}\lim_{T \rightarrow \infty}\frac1T \sum_{t=1}^T \mathbb{E}[\mathcal{J}(\clm{x}{t}(w_{t:0}),\clm{u}{t}(w_{t:0})) ]
 \\
\label{eq:const-lqrclm-constclm2} &s.t. \,\,\, \,\,\,\,\,\,\,\, \clm{x}{t}(w_{t:0}) = \clm{x}{t}(0,w_{t-1:0}) + w_{t} \\
\nonumber&\,\,\,\,\,\,\,\,\,\clm{x}{t+1}(0,w_{t:0}) = A\clm{x}{t}(w_{t:0})+B\clm{u}{t}(w_{t:0})\\
 \label{eq:const-lqrclm-constx2} &\qquad \forall t, |w_t| \leq \eta_{max}:\quad |\clm{x}{t}(w_{t:0})|\leq x_{max} \\
  \label{eq:const-lqrclm-constu2} &\qquad \forall t, |w_t| \leq \eta_{max}:\quad |\clm{u}{t}(w_{t:0})| \leq u_{max} 
\end{align}
\end{subequations}

As in the linear SLS case \cite{anderson2019system}, we do not need to have the controller $\b{K}$ be a decision variable, since we can always realize the optimal solution $(\b{\Psi}^{(x)*},\b{\Psi}^{(u)*})$ to \eqref{eq:const-lqrclm} with a system level controller $\SL{\b{\Psi}^{(x)*}}{\b{\Psi}^{(u)*}}$.
\subsection{Conservativeness of Linear Solutions}\label{sec:linneg}
We will first discuss properties of solutions to our original problem \eqref{eq:const-lqr}, if we restrict ourselves to only LTI controllers $\b{K}$. Consider the equivalent problem formulation \eqref{eq:const-lqrclm} with the CLMs $(\CLM{x},\CLM{u})$ restricted to be linear. This poses a convex problem and as shown in \cite{SLSChenAnderson2019}, it can be approximately solved by searching over FIR CLMs $(\CLM{x},\CLM{u})$ with large enough horizon $T$. Yet, the corresponding linear optimal CLMs $(\CLM{x,lin*},\CLM{u,lin*})$ come with undesirable restrictions:

\begin{itemize}
    \item $(\CLM{x,lin*},\CLM{u,lin*})$ impose stricter safety constraints than the required constraints \eqref{eq:const-lqrclm-constx2} and \eqref{eq:const-lqrclm-constu2}. 
    \item $(\CLM{x,lin*},\CLM{u,lin*})$ do not depend on the disturbance distribution $p(w)$. 
\end{itemize}

To see the first point, we have the following result as a consequence of linearity:
\begin{lem}
    \label{lem:restrict}
For any linear $(\CLM{x,lin},\CLM{u,lin})$, the constraint \eqref{eq:const-lqrclm-constx2},\eqref{eq:const-lqrclm-constu2} is equivalent to
\begin{subequations}
\label{eq:strict}
\begin{align}
\sup_t |\clm{x,lin}{t}(w_{t:0})| &\leq \sup_t  \frac{x_{max}}{\eta_{max}} |w_t|\\
\sup_t |\clm{u,lin}{t}(w_{t:0})| &\leq \sup_t  \frac{u_{max}}{\eta_{max}}|w_t|.
\end{align}
\end{subequations}
\end{lem}
\begin{proof}
Clearly, \eqref{eq:strict} implies \eqref{eq:const-lqrclm-constx2},\eqref{eq:const-lqrclm-constu2}. The reverse implication follows by the assumed linearity of $(\CLM{x,lin},\CLM{u,lin})$ and homogeneity of norms.
\end{proof}
\lemref{lem:restrict} shows that the restriction of linearity in CLMs imposes stricter safety conditions \eqref{eq:strict} than \eqref{eq:const-lqrclm-constx2},\eqref{eq:const-lqrclm-constu2}. To elaborate on the second point, notice that for linear CLMs $(\CLM{x,lin},\CLM{u,lin})$, the objective function \eqref{eq:cost-lqrclm-cost2} can be expressed equivalently as 
\begin{align}\label{eq:H2simplify}
    \eqref{eq:cost-lqrclm-cost2} = \sigma^2\left\|\begin{array}{c}Q^{1/2} \CLM{x}\\ P^{1/2}\CLM{u} \end{array} \right\|^2_{\mathcal{H}_2},\quad \sigma^2:=\E_{w \sim p(w)}[w^2]
\end{align}
where $\sigma^2$ denotes the variance of the scalar distribution $p(w)$ and $\|.\|_{\mathcal{H}_2}$ denotes the $\mathcal{H}_2$ norm for linear operators. Since the objective function only gets scaled by a constant factor $\sigma^2$ for different distributions $p(w)$, this shows that for linear CLMs, the solutions $(\CLM{x,lin},\CLM{u,lin})$ to \eqref{eq:const-lqrclm} are independent of the distribution $p(w)$. 

\subsection{A Nonlinear System Level Approach}
\label{sec:nlsls-take}
Consider the general problem \eqref{eq:const-lqrclm}, where now we search over CLMs $(\CLM{x},\CLM{u})$ of the form presented in \eqref{eq:NL-map} with the choice of $\eta_2= \eta_{max}$, some $\eta_1<\eta_2$, and an FIR horizon $T$. Recall that $(\CLM{x},\CLM{u})$ is a blending of two linear CLMs and has the form \eqref{eq:NL-map}. Restricting ourselves to this form of CLMs allows to derive the following convex problem which is a relaxation of the general problem  \eqref{eq:const-lqrclm}:
\begin{subequations}
\label{eq:S2}
\begin{align}
 \label{eq:cost-blend} \min_{\R{i}{}, \M{i}{}} & \left\|\begin{bmatrix} Q&0\\ 0 & P\end{bmatrix}^{1/2} \begin{bmatrix} \R1{} & \R2{} \\ \M1{}&\M2{}\end{bmatrix} \Sigma^{1/2}_w \right\|^{2}_{F}  \\
\label{eq:safety-1} 
s.t.\quad& \eta_1|\R1{}| + (\eta_2 - \eta_1)|\R2{}| \leq x_{max}\\
 &\eta_1|\M1{}| + (\eta_2 - \eta_1)|\M2{}| \leq u_{max} \label{eq:safety-2}\\
 & \R{i}{k+1} = A\R{i}{k} + B\M{i}k \label{lin-feas1}\\
 & \R{i}{1} = I,\quad \R{i}{T} = 0  \nonumber
\end{align}
\end{subequations}
where 
\begin{equation}
\label{eq:sigma} 
      \Sigma_w = \begin{bmatrix} \alpha_1 I & \alpha_2 I \\ \alpha_2 I &  \alpha_3 I  \end{bmatrix} \nonumber
\end{equation}
with $\alpha_1 = \E[\p1{w}^2]$, $\alpha_2 = \E[\p1{w} (\p2{w} -\p1{w})]$, and $\alpha_3 = \E[(\p2{w} - \p1{w})^2]$, where $w \sim p(w)$ and $\|\b{w}\|_{\infty} \leq \eta_{max} $. Moreover $\R{i}{}$ and $\M{i}{}$ are abbreviations for the row-wise concatenation of the matrices associated with the linear CLMs $\CLM{x,i}$, $\CLM{u,i}$, \textit{i.e}, $\R{i}{} = [\R{i}{T}, \R{i}{T-1}, \dots ,\R{i}{1}]$, $\M{i}{} = [\M{i}{T}, \M{i}{T-1}, \dots ,\M{i}{1}]$. Hereby, only constraints \eqref{eq:safety-1}, \eqref{eq:safety-2} are sufficient condition of the constraint \eqref{eq:const-lqrclm-constx2}, \eqref{eq:const-lqrclm-constu2} via norm multiplicativity. All other equations in the above optimization are equivalent to the original problem \eqref{eq:const-lqrclm} restricting the search over CLMs of the form \eqref{eq:NL-map}. Finally, solving the convex problem \eqref{eq:S2} gives the sub-optimal nonlinear CLMs $(\CLM{*x}, \CLM{*u})$ for the system dynamics \eqref{eq:lin-dyn}, realized by an internally stabilizing controller $\SL{\CLM{*x}}{\CLM{*u}}$.
The next theorem states a main result of this paper:
\begin{thm}\label{thm:nlslsvssls}
For all $\eta_1 \in [0,\eta_2]$, the nonlinear system level controller $\SL{\CLM{*x}}{\CLM{*u}}$ synthesized from \eqref{eq:S2} achieves lower optimal LQR cost for \eqref{eq:const-lqr} than any linear solutions.
\end{thm}
\begin{proof}
First, recall that restricting $\b{K}$ to be linear in problem \eqref{eq:const-lqr} is equivalent to restricting $\CLM{x}$ and $\CLM{u}$ to be linear in the equivalent formulation \eqref{eq:const-lqrclm}. Furthermore, notice that under the restriction of linear $(\CLM{x}, \CLM{u})$, problem \eqref{eq:const-lqrclm} is equivalent to \eqref{eq:S2} with the added constraint $\R1{} = \R2{}$, $\M1{}=\M2{}$, which shows that any solution $(\CLM{*x}, \CLM{*u})$ of problem \eqref{eq:S2} achieve smaller cost than a linear solution $(\CLM{x,lin*}, \CLM{u,lin*})$ of \eqref{eq:const-lqrclm}.
\end{proof}
\begin{rem}
The above argument extends directly to the N-blend case.
\end{rem}

\subsection{Localized Controller for Constrained LQR}
\label{sec:ls-lqr}
 Thanks to the particular form of \eqref{eq:n-blend}, when the projection is chosen to be the saturation projection \defref{def:sat}, structural constraints of controller such as sparsity and delay constraints can be added in a convex way to the synthesis procedure described in \secref{sec:nlsls-take}. This is because imposing structural constraints on the nonlinear controller \eqref{eq:n-blend} is equivalent to imposing them on the linear CLM components of \eqref{eq:n-blend}. Detailed in \cite{anderson2019system}, localization of disturbance, communication and actuation delay, as well as sparsity pattern are all convex constraints in terms of linear CLMs in the linear System Level Synthesis framework. Specifically, all mentioned constraints could be cast as a convex subspace $\mathcal{S}_x$ and $\mathcal{S}_u$ for linear CLMs $\CLM{x,i}$,$\CLM{u,i}$,$i \in \set{N}$. The corresponding system level controller $\SL{\TU}{\TX}$ can then be implemented in a localized fashion conforming to the subspace constraints on  $\CLM{x,i}$,$\CLM{u,i}$. Therefore, the nonlinear controller synthesis in \secref{sec:nlsls-take} naturally inherits all capabilities of the linear system level controllers in terms of distributed controller synthesis and implementation.

\section{Distributed Anti-windup Controller for Saturated Systems}
\label{sec:ls}
Now consider a linear input saturated system where the disturbances and initial condition are \textit{not} necessarily constrained to have a known norm bound $\eta_{max}$. The control actions are projected via saturation projector:
\begin{align}
\label{eq:lin-sat}
x_{t} = Ax_{t-1}+ B\textit{P}_{u_{max}}(u_{t-1}) + w_t
\end{align}
In this scenario, controller $\SL{\CLM{x}}{\CLM{u}}$ previously constructed with \eqref{eq:NL-map} \textit{no longer} realizes the designed closed-loop response \eqref{eq:NL-map} for \eqref{eq:lin-sat}. Nevertheless,  we  would  like  the  saturated  system  to  degrade  gracefully  and  preserve  stability.  Such  property  is traditionally  achieved  via  anti-windup  design\cite{hu2008anti}. Here, we show that the proposed nonlinear controller achieves natural anti-windup property with little modification.

\subsection{Anti-windup Controller}
Inspired by internal model control (IMC) \cite{zheng1994anti}, we modify  $\SL{\CLM{x}}{\CLM{u}}$ and consider an augmented controller $\SL{\CLM{x,a}}{\CLM{u}}$ where the operator $\CLM{x,a}$ is constructed from \eqref{eq:NL-map} with augmentation:
\begin{align}
    \label{eq:augment}
    \CLM{x,a}_t(w_{t:0}) =\sum^{N}_{i=1} \Bigg( &\sum^{\min\{T,t+1\}}_{k=1} R^{(i)}_{k}  (\pp{i}-\pp{i-1})(w_{t+1-k}) \Bigg) \nonumber\\ 
        \,\,\,\,\, +  &\sum_{k = 1}^{\tau+1} A^{k-1} \left(w_{t+1-k}-\pp{N}(w_{t-k+1})  \right),
\end{align}
where $\tau$ is a positive integer. Recall that by design, we have chosen $\eta_{N} = \eta_{max}$, the expected norm bound on disturbances. Compared to \eqref{eq:NL-map} for the two-zone case ($N=2$) and \eqref{eq:n-blend} for the N-zone case, we note that \eqref{eq:augment} has the additional "open-loop" dynamics term. This extra term accounts for the residual disturbances that are not attenuated by the original controller $\SL{\CLM{x}}{\CLM{u}}$ because the disturbances are larger than expected by the projection mapping, \textit{i.e.}, $|w_t| > \eta_{max}$. Therefore, $\SL{\CLM{x,a}}{\CLM{u}}$ considers the $\tau$-step propagation of the unaccounted disturbances from $\SL{\CLM{x}}{\CLM{u}}$. Note that when the disturbances satisfy the assumption $\|\b{w}\|_{\infty} \leq \eta_{max}$, augmented controller $\SL{\CLM{x,a}}{\CLM{u}}$ is identical to $\SL{\CLM{x}}{\CLM{u}}$ constructed from \eqref{eq:n-blend} and \eqref{eq:NL-map}.

The IMC-like structure of the augmented controller $\SL{\CLM{x,a}}{\CLM{u}}$ helps the saturated system to degrade gracefully and preserve stability even when $\CLM{x,a},\CLM{u}$ are not the exact CLMs for the closed-loop system. The closed-loop dynamics of \eqref{eq:lin-sat} under augmented controller $\SL{\CLM{x,a}}{\CLM{u}}$ from \eqref{eq:augment} can be checked to be:
\begin{equation}
    \label{internal}
    \hat{w}_t = A^{\tau+1}(\hat{w}_{t-\tau}-\p{max}{\hat{w}_{t-\tau}}) + w_{t}.
\end{equation}
As shown in \cite{NLSLSHODoyle2019}, the stability of the overall closed loop is equivalent to the stability of \eqref{internal}. We now certify the anti-windup property of $\SL{\CLM{x,a}}{\CLM{u}}$ with the following result.
\begin{lem}
    \label{lem:transient}
    If $\tau$ satisfies $|A^{\tau+1}|<1$, then internal dynamics \eqref{internal} is globally finite-gain $\ell_{\infty}$-stable where for all $\b{w} \in \ell_{\infty}^n$,
    $$\|\hat{\b{w}}\|_{\infty} \leq \frac{1}{1-|A^{\tau+1}|} \|\b{w}\|_{\infty} $$ 
\end{lem}
\begin{proof}
    See Appendix.
\end{proof}
In particular, if $A$ is schur, then there exists $k\in \mathbb{N}$ such that $\|A^k\|<1$ for any norm. Therefore, if \eqref{eq:lin-sat} is open-loop stable, $\SL{\CLM{x,a}}{\CLM{u}}$ guarantees graceful degradation when the closed-loop is saturated. 

\subsection{Localized Implementation}
Similar to the large-scale constrained LQR case in \secref{sec:lqr}, since the anti-windup controller $\SL{\CLM{x,a}}{\CLM{u}}$ for the saturated linear system \eqref{eq:lin-sat} is composed of linear CLMs synthesized from \eqref{eq:S2} with locality constraints, localization can be easily imposed as a convex subspace constraint on the composing linear CLMs. When the the information structure of the controllers are constrained to the state propagation pattern according to open-loop dynamics \textit{i.e.}, the sparsity of $A$, the anti-windup controller $\SL{\CLM{x,a}}{\CLM{u}}$ can be implemented in a localized fashion where information is exchanged and disturbance is contained in a local controller patch \cite{wang2014localized}. As will be illustrated in \secref{sec:sim}, this allows for distributed anti-windup controller design for large-scale saturated systems.

\section{Simulation}
\label{sec:sim}
\subsection{Constrained LQR}
To corroborate the results presented in the previous sections, we demonstrate the performance of a four-zone nonlinear blending controller with radial projection compared against the optimal linear controller for the constrained LQR problem of an open-loop unstable system: 
\begin{align}
    x_{t} &= \begin{bmatrix}1&1&0\\1&2&1\\0&1&1 \end{bmatrix}x_{t-1} + \begin{bmatrix}0\\0\\1 \end{bmatrix}u_{t-1} +w_{t}
\end{align}
with $u_{max} = 40$, $x_{max} = 15$, $\eta_{max} = 1$, $Q = I_3$, $P = 10$. The disturbances $w_k$ are chosen to be a truncated i.i.d. gaussian random variables with variance $\sigma^2$. \figref{fig:performance} shows the optimal cost improvement of the presented nonlinear approach over the optimal linear controller for different choices of variance $\sigma^2$. \figref{fig:performance} showcases that the proposed controller can exploit the knowledge of the disturbance distribution to achieve performance improvement over the linear optimal linear controller: For small $\sigma$ the proposed controller gains more than 30\% cost reduction over safe controller. On the other hand, with increasing $\sigma$, large disturbances in the system become more likely, and therefore the opportunity to improve upon the linear optimal controller is reduced.
\begin{figure}[h]
 \centering
 \includegraphics[scale = 0.8]{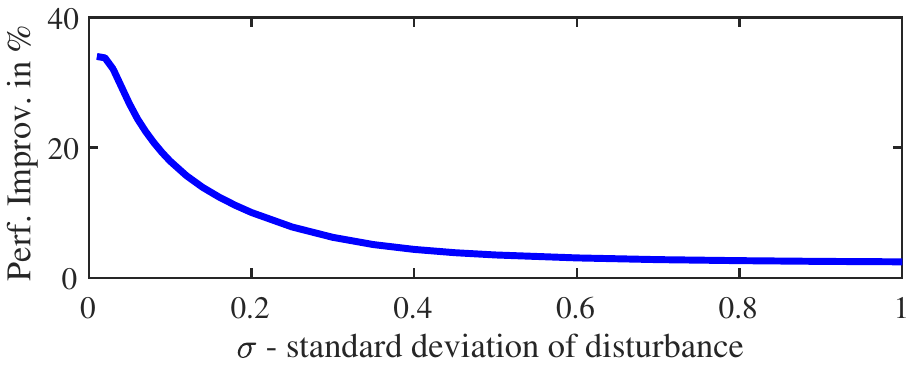}
 \caption{Performance improvement of optimal nonlinear controller $\SL{\CLM{*x}}{\CLM{*u}}$ over optimal linear controller $\SL{\CLM{x,lin*}}{\CLM{u,lin*}}$ for different variances $\sigma^2$ of the non-truncated disturbance. The nonlinear blending controller synthesizes over 4 linear controllers w.r.t. to the projection parameters $\eta_1= 0.05, \eta_2 = 0.1, \eta_3 = 0.2, \eta_4 = \eta_{max} = 1 $}
 \label{fig:performance}
\end{figure}

\subsection{Localized Anti-Windup Controller}
Consider a bi-directional chain system with $i$th node's dynamics being
\begin{equation}
\label{eqn:sim}
    x^i_{t+1} = (1- 0.4\lvert \mathcal{N}(i) \rvert  ) x^i_t + 0.4 \sum_{j \in \mathcal{N}_i} x^j_t + \sat(u^i_t,u_{max}) + w_{t}^i \nonumber
\end{equation}
where $\mathcal{N}(i)$ denotes the set of vertices that has an edge connected to $i$th vertex and $w_t^i$ is the $i$th coordinate of disturbance vector at time $t$. In particular, $\|\b{w}\|_{\infty}\leq 1$ and $x_0 = 0$.  One can check that the overall chain system is open-loop marginally (un)stable.

We illustrate the anti-windup property of the nonlinear controller \eqref{eq:n-blend} in the decentralized setting with additional sparsity, locality, and delay constraints in  \figref{fig:ls-worst}. First, a nominal integral controller for this system is designed and dubbed as the \textit{Integral Controller}. Due to its integral structure, the \textit{Integral Controller} for the unconstrained closed loop guarantees convergence of the state to the origin under persistent disturbance, \textit{i.e}, step rejection. In comparison, a second linear controller synthesized from standard constrained LQR problem that guarantees stability for all admissible $\b{w}$ under saturation is generated. We refer to this linear controller as the \textit{Non-integral Controller} since the states only stay bounded under persistent admissable disturbance.  

 The nonlinear controller with saturation projection here is chosen to be a two-zone blending controller consisted of CLMs of the form \eqref{eq:NL-map}. The simulation shows the anti-windup property as well as preservation of step rejection in both large- and small-disturbance schemes of the proposed method. \figref{fig:ls-worst} shows that the blending controller stabilizes the system while integral controller becomes unstable under worst-case bounded disturbance. On the other hand, the proposed blending controller preserves performance of step rejection while the linear Non-integral Controllers forfeits the performance objective in order to preserve stability in the saturated closed loop. In this chain example, we allow 1 time step communication delay between nodes and actuation delay with 50\% control authority. The localization pattern imposed on the system response allows $\SL{\CLM{x}}{ \CLM{u}}$ to be implemented in local patches, therefore making the controller distributed.

\begin{figure*}
\centering
 \subfigure[Integral Controller]{ \includegraphics{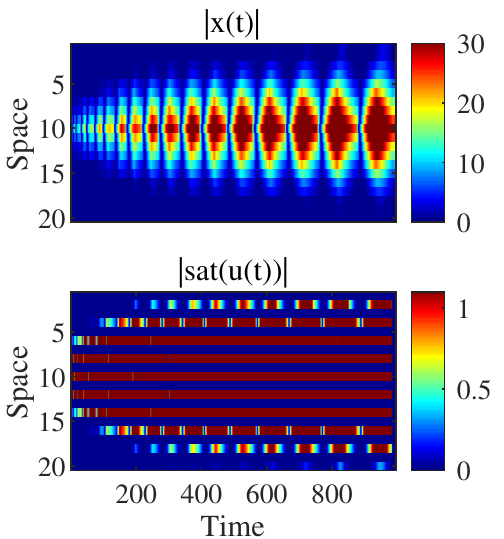}}
  \subfigure[Non-integral Controller]{ \includegraphics{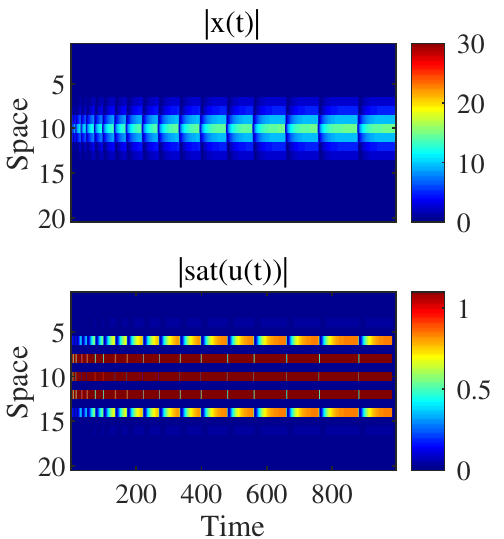}\label{fig:ls-act}}
    \subfigure[Nonlinear Blending Controller ]{ \includegraphics{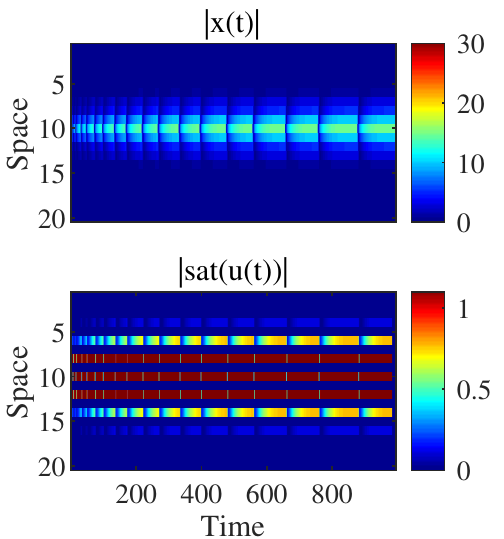}\label{fig:ls-act}}
     \subfigure[Integral Controller ]{ \includegraphics{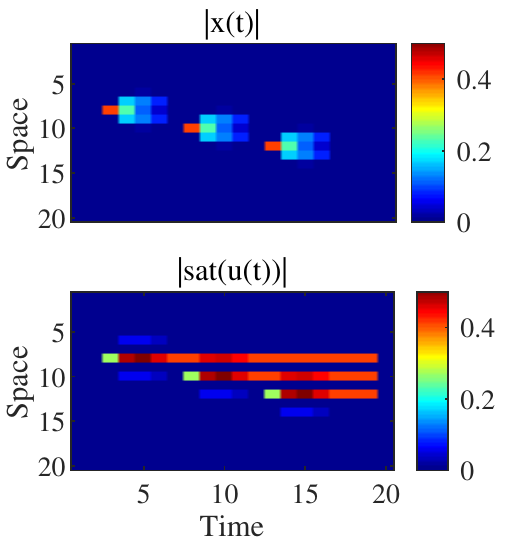}}   
  \subfigure[Non-integral Controller ]{ \includegraphics{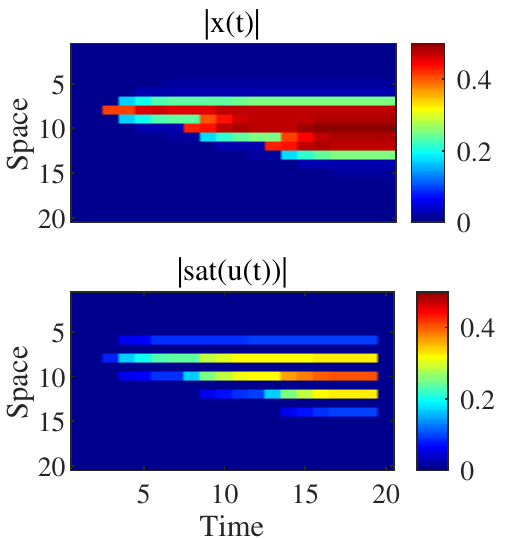}\label{fig:ls-act}}
    \subfigure[Nonlinear Blending Controller ]{ \includegraphics{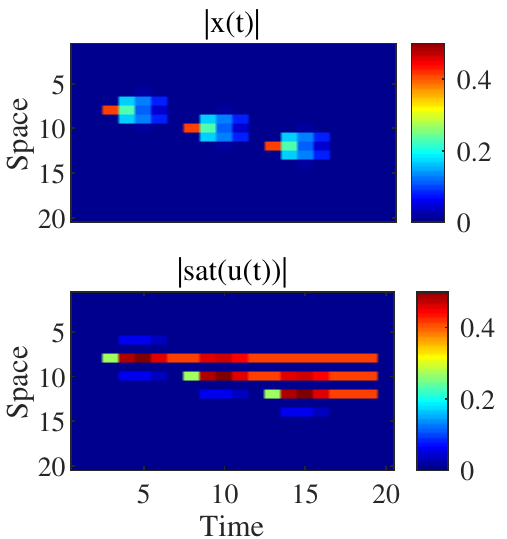}\label{fig:ls-act}}
    \caption{\textit{\textbf{Worst-case (top row) and Staggered Step Input (bottom row) Response} under Saturation for a marginally stable 20-node Chain System with Sparse Actuation}: \textit{Top Row}: The heatmaps show how a worst-case disturbance is propagated through space-time for the saturated chain system. The integral controller becomes unstable due to saturation and the naive blending controller possesses has the anti-windup property of the non-integral controller. In addition to  anti-windup, the proposed controller is localized and accommodates sparse actuation, communication delay, and controller sparsity constraints. Here every other node has a control input (50\% actuation) with 1 time step actuation delay and 1 time step communication delay between nodes, while enforcing a controller sparsity that conforms to the communication pattern of dynamics matrix $A$. \textit{Bottom Row}: Response to small step disturbances at node 8,10,12 entering at time 2,6,10, respectively. As in the scalar case, the proposed blending controller not only stabilizes under saturation but also recovers the performance objective of rejecting small step disturbances. This contrasts against the non-integral controller, which sacrifices small-signal performance for stability. }
   \label{fig:ls-worst}
\end{figure*} 

\section{conclusion}
\label{sec:conclusion}
    We showcase the nonlinear system level approach developed in \cite{NLSLSHODoyle2019} and illustrate the use cases for a class of nonlinear system level controllers. We propose a tractable nonlinear control synthesis method that outperforms any optimal linear controller for the constrained LQR problems. It was further shown that such controller naturally possesses anti-windup property for linear systems with input saturation. A key highlight is that the presented approach enjoys the same compatibility with locality/ delay constraints and distributed implementation, as the linear system level approach\cite{anderson2019system}. Overall, this paper is a first step in exploring the full potential of the new nonlinear control synthesis framework developed in \cite{NLSLSHODoyle2019} and highlights that even just the presented special case of the framework, called "nonlinear blending" of linear controllers, offers many benefits.
    

\bibliographystyle{IEEEtran}
\bibliography{refs}

\appendix
\section{Proof of Lemma VI.1}
We first present an operator small-gain theorem.
\begin{thm}[Small-gain Theorem\cite{NLSLSHODoyle2019}]
    \label{thrm}
    Let $\Delta \in \scausal{n}{n}$. If for all $\b{x} \in \ell_p^n$, $\|\Delta(\b{x})\|_p \leq \gamma \|\b{x}\|_p +\beta$ with $0<\gamma<1$, $\beta \geq 0$, $p = 1,2,\dots, \infty$, then for all $\b{w} \in \ell_p^n$, $ \|\hat{\b{w}}\|_p \leq \frac{1}{1-\gamma} (\|\b{w}\|_p +\beta)$ where $\hat{\b{w}} = (I-\Delta)^{-1}\b{w}$.
\end{thm}
Note that the inverse exists because $\Delta \in \causal{n}{n}$ \cite{NLSLSHODoyle2019}. We are now in a position to prove \lemref{lem:transient}. We can write \eqref{internal} in the operator form as 
\begin{equation}
    \label{eqn:cl}
    \hat{\b{w}} = (I-\Delta)^{-1}\b{w},
\end{equation}
where $\Delta$ is a strictly causal operator with component function $\Delta_t(\hat{w}_{t:0}) := A^{\tau+1}(\hat{w}_{t-\tau}-\p{max}{\hat{w}_{t-\tau}})$. For all $\hat{\b{w}} \in \ell_{\infty}^n$, $\|\Delta(\hat{\b{w}} )\|_{\infty} \leq |A^{\tau+1}| \|\hat{\b{w}} \|_{\infty}$ where we have chosen $\tau$ such that $|A^{\tau +1}|<1$. Therefore, invoking \thmref{thrm} gives the desired result in \lemref{lem:transient}.

\end{document}